\documentclass[3p]{elsarticle}

\usepackage{verbatim}
\usepackage[usenames]{color}
\usepackage{amssymb,amsfonts,amsmath,enumerate,latexsym,yhmath,amsthm}
\usepackage{url}
\usepackage{graphicx}

\usepackage{tikz}
\usetikzlibrary{snakes,shapes}
\tikzstyle{hyb}=[rectangle,fill=green!50,draw,minimum size=3mm]
\tikzstyle{tre}=[circle,fill=green!50,draw,minimum size=3.5mm]

 \theoremstyle{plain}
 \newtheorem{theorem}{Theorem}

 \newtheorem{proposition}[theorem]{Proposition}

 \newtheorem*{nlemma}{Lemma}
 
 \newtheorem*{nproposition}{Proposition}
 \theoremstyle{definition}

\renewcommand{\leq}{\leqslant}
\renewcommand{\geq}{\geqslant}

\newcommand{\TT}{\mathcal{T}}

\newcommand{\BT}{\mathcal{T}}

\begin{document}
\begin{frontmatter}

\title{The expected value  under the Yule model of the squared path-difference distance}

\author{Gabriel Cardona}
\ead{gabriel.cardona@uib.es}
\author{Arnau Mir}
\ead{arnau.mir@uib.es}
\author{Francesc Rossell\'o\corref{cor1}}
\ead{cesc.rossello@uib.es}
\cortext[cor1]{Corresponding author}
\address{Department of Mathematics and  Computer Science,
  University of the Balearic Islands,
  E-07122 Palma de
  Mallorca, Spain}

\begin{abstract}
The path-difference metric is one of the oldest and most popular distances for the comparison of phylogenetic trees, but its statistical properties are still quite unknown. In this paper we compute the expected value  under the Yule model of evolution of its square on the space of fully resolved rooted phylogenetic trees with $n$ leaves. This complements previous work by Steel--Penny and Mir--Rossell\'o,   who computed this mean value  for fully resolved unrooted and rooted phylogenetic trees, respectively, under the uniform distribution.
 \end{abstract}

\end{frontmatter}

\section{Introduction}
\label{sec:intro}

The definition and study of metrics for the comparison of rooted phylogenetic trees  on the same set of taxa is a classical problem in phylogenetics \cite[Ch.~30]{fel:04}.  A classical and popular family of such metrics is based on the comparison, by different methods, of the vectors of lengths of the (undirected) paths connecting all pairs of taxa in the corresponding trees  \cite{farris:69,farris:73,phipps:71,willcliff:71}. These metrics are generically called \emph{nodal distances}, although some of them have also specific names. For instance, 
 the metric defined through the euclidean distance between path-lengths vectors is called  \emph{path-difference metric} \cite{steelpenny:sb93}, or  \emph{cladistic difference}  \cite{farris:69}. 
 
In contrast with other metrics, the statistical properties of these nodal distances are mostly unknown. Actually, the only statistical property that has been established so far for any one of them is the expected, or mean, value  of the square of the path-difference metric for unrooted  \cite{steelpenny:sb93}  and rooted \cite{MirR10} fully resolved phylogenetic trees under the uniform distribution (that is, when all phylogenetic trees with the same number of taxa are equiprobable). The knowledge of the expected value of a metric  is useful, because it  provides an indication about the significance of the similarity of two individuals measured through this metric  \cite{steelpenny:sb93}. 

But phylogeneticists consider also other probabilistic distributions on the space of phylogenetic trees on a fixed set of taxa, defined through stochastic models of evolution \cite[Ch.~33]{fel:04}. The most popular such model is Yule's  \cite{Harding71,Yule}, defined by an evolutionary process where, at each step, each currently extant species can give rise,  with the same probability, to two new species. Under this model, different phylogenetic trees with the same number of leaves may have different probabilities. Formal details on this model are given in the next section.

In this paper we compute the expected value of the square of the path-difference metric for rooted fully resolved phylogenetic trees under this Yule model. Besides the aforementioned application of this value in the assessment of tree comparisons, the knowledge of formulas for this expected value under different models may allow the use of the path-difference metric to test stochastic models  of tree growth, a popular line of research in the last years which so far has been mostly based on shape indices \cite{Mooers97}.

\section{Preliminaries}
\label{sec:prel}

In this paper, by a  \emph{phylogenetic tree} on a set $S$ of taxa we mean a fully resolved, or binary, rooted tree  with its leaves bijectively labeled in $S$.  We understand such a rooted  tree as a directed graph, with its arcs pointing away from the root. To simplify the language, we shall always identify a leaf of a phylogenetic tree with its label.  We shall also use the term \emph{phylogenetic tree with $n$ leaves} to refer to a phylogenetic tree on the set $\{1,\ldots,n\}$.  We shall denote by $\TT(S)$ the space of all phylogenetic trees on $S$ and by 
$\TT_n$ the space of all phylogenetic trees with $n$ leaves.

Whenever there exists a directed path from $u$ to $v$ in a phylogenetic tree $T$, we shall say that $v$ is a  \emph{descendant} of $u$.  The \emph{distance} $d_T(u,v)$ between two nodes $u,v$ in a phylogenetic tree $T$ is the length  (in number of arcs) of the unique undirected path connecting $u$ and $v$. The \emph{depth} $\delta_T(v)$ of a node $v$  in $T$ is the distance from the root $r$ of $T$ to $v$.  The \emph{path-difference distance} \cite{farris:69,farris:73} between a pair of trees $T,T'\in \TT_n$ is
$$
d_{\nu}(T,T')=\sqrt{\sum_{1\leq i<j\leq n} (d_T(i,j)-d_{T'}(i,j))^2}.
$$

The \emph{Yule}, or  \emph{Equal-Rate Markov}, model of evolution \cite{Harding71,Yule} is a stochastic model of phylogenetic trees' growth. It starts with a node, and at every step a leaf is chosen randomly and uniformly and it is splitted into two leaves. Finally, the labels are assigned randomly and uniformly to the leaves once the desired number of leaves is reached. Under this model, if $T$ is a phylogenetic{} tree with $n$ leaves and set of internal nodes $V_{int}(T)$, and if for every internal node $v$ we denote by $\ell_T(v)$ the number of its descendant leaves, then the probability of  $T$ is   \cite{Brown,SM01}
 $$
 P_Y(T)=\frac{2^{n-1}}{n!}\prod_{v\in V_{int}(T)}\frac{1}{\ell_T(v)-1}.
$$

For every $n\geq 1$, let $H_n=\sum_{i=1}^n 1/i$ and $H_n^{(2)}=\sum_{i=1}^n 1/i^2$. Let, moreover, $H_0=H_0^{(2)}=0$. $H_n$ is called the $n$-th \emph{harmonic number}, and $H_n^{(2)}$, the $n$-th \emph{generalized harmonic number of power $2$}. 

\section{Main results}

Let $N_n^2$ the random variable that chooses independently a pair of trees $T,T'\in \BT_n$ and computes 
$$
d_{\nu}(T,T')^2=\sum_{1\leq i<j\leq n} (d_T(i,j)-d_{T'}(i,j))^2.
$$
In this section we establish the following result.

\begin{theorem}\label{th:main}
The  expected value of $N_n^2$ under the Yule model is
$$
E_Y(N_n^2)=\frac{2n}{n-1}\big(2(n^2+24n+7)H_n+13n^2-46n+1-16(n+1)H_n^2-8(n^2-1)H_n^{(2)}\big).
$$
\end{theorem}

To prove this formula, we shall use the following auxiliary random variables:
\begin{itemize}

\item $D_n$ is the random variable that chooses a tree $T\in \BT_n$ and computes
$D (T)=\hspace*{-2ex} \sum\limits_{1\leq i< j\leq n}d_T(i,j)$.

\item $D_n^{(2)}$ is the random variable that chooses a tree $T\in \BT_n$ and computes
$D^{(2)}(T)=\hspace*{-2ex}\sum\limits_{1\leq i< j\leq n}d_T(i,j)^2$.
\end{itemize}

The connection between $E_Y(N_n^2)$ and the expected values under the Yule model of $D_n,D_n^{(2)}$ is given by the following result.

\begin{proposition}\label{prop:EN2}
$E_Y(N_n^2)=2\big(E_Y(D^{(2)}_n)- E_Y(D_n)^2/\binom{n}{2}\big).$
\end{proposition}

\begin{proof}
Let us develop $E_Y(N_n^2)$ from its raw definition:
$$
\begin{array}{l}
\displaystyle E_Y(N_n^2)=\hspace*{-1ex}\sum_{T,T'\in \BT_n} d_{\nu}(T,T')^2p_Y(T)p_Y(T')=\hspace*{-1ex}\sum_{T,T'\in \BT_n} \Big(\sum_{1\leq i< j\leq n} (d_T(i,j)-d_{T'}(i,j))^2\Big)p_Y(T)p_Y(T')\\
\quad \displaystyle 
=\hspace*{-1ex}\sum_{1\leq i < j\leq n}\sum_{T,T'} (d_T(i,j)^2+d_{T'}(i,j)^2-2d_T(i,j)d_{T'}(i,j))p_Y(T)p_Y(T')
\\
\quad \displaystyle 
=\hspace*{-1ex}\sum_{1\leq i< j\leq n}\Big(\sum_{T,T'}d_T(i,j)^2 p_Y(T)p_Y(T')+\sum_{T,T'} d_{T'}(i,j)^2p_Y(T)p_Y(T')\\
\quad\qquad\qquad\qquad \displaystyle
-2  \sum_{T,T'} d_T(i,j)d_{T'}(i,j)p_Y(T)p_Y(T')\Big)\\
\quad \displaystyle 
=\hspace*{-1ex}\sum_{1\leq i< j\leq n}\hspace*{-1ex}\Big(\sum_{T}d_T(i,j)^2 p_Y(T)+\sum_{T'} d_{T'}(i,j)^2p_Y(T')-2\Big(\sum_{T}d_T(i,j) p_Y(T)\Big)\Big(\sum_{T'} d_{T'}(i,j) p_Y(T')\Big)\Big)\\
\quad \displaystyle 
=\hspace*{-1ex}\sum_{1\leq i< j\leq n}\Big(2\sum_{T}d_T(i,j)^2 p_Y(T)
-2\Big(\sum_{T}d_T(i,j) p_Y(T)\Big)^2\Big)\\
\quad \displaystyle 
=2\sum_{T} \Big(\sum_{1\leq i< j\leq n} d_T(i,j)^2\Big) p_Y(T) -2\sum_{1\leq i< j\leq n}
\Big(\sum_{T}d_T(i,j) p_Y(T)\Big)^2\\
\quad \displaystyle =
2E_Y(D_n^{(2)})-2 \binom{n}{2}\Big(\sum_{T}{d}_T(1,2) p_Y(T)\Big)^2
\end{array}
$$
and now
$$
E_Y(D_n) =\sum_{T\in \BT_n} \sum_{1\leq i<j\leq n}d_T(i,j)p_Y(T)= \sum_{1\leq i<j\leq n}\sum_{T} d_T(i,j)p_Y(T) =
\binom{n}{2}\sum_{T} d_T(1,2)p_Y(T)
$$
from where we deduce that 
$\Big(\sum\limits_{T}{d}_T(1,2) p_Y(T)\Big)^2= {E_Y(D_n)^2}/{\binom{n}{2}^2}$,
and the formula in the statement follows.
\end{proof}

Now, it is known that the expected value under the Yule model of $D_n$ is  
$$
E_Y(D_n)=2n(n+1)H_n-4n^2\qquad \mbox{\cite{MRR}}.
$$  
As far as $E_Y(D_n^{(2)})$ goes, its value is given by the following result. We postpone the proof until the appendix at the end of the paper. 

\begin{theorem}\label{EY:D2}
$E_Y(D_n^{(2)})= 8n(n+1)(H_{n}^2-H_{n}^{(2)})-2n(15n+7)H_n+45n^2-n$
\end{theorem}

Then, replacing in the expression for $E_Y(N_n^2)$ given in Proposition \ref{prop:EN2} , $E_Y(D_n)$ and $E_Y(D_n^{(2)})$ by their values, we obtain the formula for $E_Y(N_n^2)$ given in Theorem \ref{th:main}.

\section{Conclusions}

In this paper we have computed the expected value $E_Y(N_n^2)$ of the square of the path-difference metric for rooted fully resolved phylogenetic trees under the Yule model:
$$
E_Y(N_n^2)=\frac{2n}{n-1}\big(2(n^2+24n+7)H_n+13n^2-46n+1-16(n+1)H_n^2-8(n^2-1)H_n^{(2)}\big).
$$
This complements the computation of this expected value under the uniform distribution carried out in   \cite{MirR10}, which turned out to be
$$
E_U(N_n^2)=2\binom{n}{2}\left(4(n-1)+2-\frac{2^{2(n-1)}}{\binom{2(n-1)}{n-1}}-\left(\frac{2^{2(n-1)}}{\binom{2(n-1)}{n-1}}\right)^2\right)
$$

The proof of the formula for $E_Y(N_n^2)$ consists of several long algebraic manipulations of sums of sequences. Since it is not difficult to slip some mistake in such long algebraic computations, to double-check our result we have directly computed the value of $E_Y(N_n^2)$ for $n=3,\ldots,7$ and confirmed that our formula gives the right figures. The Python scripts used to compute them and the results obtained are available in the Supplementary Material web page \url{http:/bioinfo.uib.es/~recerca/phylotrees/nodaldistYule/}.

The  formulas for $E_Y(N_n^2)$ and $E_U(N_n^2)$ grow in different orders: $E_Y(N_n^2)$ is in $O(n^2\ln(n))$, while $E_U(N_n^2)$ is in $O(n^3)$. Therefore, they can be used to test the Yule and the uniform models as null stochastic models of evolution for collections of phylogenetic trees reconstructed by different methods. This kind of analysis has only been performed so far through shape indices of single trees, not by means of the comparison of pairs of trees. We shall report on it elsewhere.

\section*{Acknowledgements}  The research reported in this paper has been partially supported by the Spanish government and the UE FEDER program, through projects MTM2009-07165 and TIN2008-04487-E/TIN.  We thank J. Mir\'o for several comments on a previous version of this work.

\section*{Appendix}

In this appendix we prove Proposition  \ref{EY:D2}, as well as of some preliminary lemmas.
To begin with, the following identities  on harmonic numbers will be systematically used in the next proofs, usually without any further notice.

\begin{nlemma}\label{lem:sumharm}
For every $n\geq 2$:
\begin{enumerate}[(1)]
\item $\displaystyle \sum_{k=1}^{n-1} H_k=n(H_{n}-1)$

\item $\sum\limits_{k=1}^{n-1} kH_k=\frac{1}{4}n(n-1)(2H_{n}-1)$

\item $\sum\limits_{k=1}^{n-1} {H_k}/({k+1})=\frac{1}{2}(H_n^2-H_n^{(2)})$

\item $\sum\limits_{k=1}^{n-1} kH_kH_{n-k}=\binom{n+1}{2}(H_{n+1}^2-H_{n+1}^{(2)}-2H_{n+1}+2)$

\item $ \sum\limits_{k=1}^{n-1} (H_k^2-H_k^{(2)})=n(H_{n}^2-H_n^{(2)})-2n(H_{n}-1)$

\item $\sum\limits_{k=1}^{n-1} k(H_k^2-H_k^{(2)})=\binom{n}{2}(H_n^2-H_n^{(2)})-\frac{1}{4}n(n-1)(2H_{n}-1)$
\end{enumerate}
\end{nlemma}

\begin{proof}
Identities (1)--(3) are well known and easily proved by induction on $n$: see, for instance,  \cite[\S 6.3, 6.4]{Knuth2} and  \cite[\S 1.2.7]{Knuth1}. Identity (4) is proved in \cite[Thm. 2]{WGW}. We shall prove (5) and (6) using the technique introduced in \cite{CHJ}. The main ingredient is \emph{Abel's lemma on summation by parts}: for every two sequences $(a_k)_k$ and $(b_k)_k$,
$$
\sum_{k=1}^{n-1}(a_{k+1}-a_k)b_k=-\sum_{k=1}^{n-1}(b_{k+1}-b_k)a_{k+1}+a_nb_n-a_1b_1.
$$
To prove (5), take $a_k=k$ and $b_k=H_k^2-H_k^{(2)}$, so that $a_{k+1}-a_k=1$ and
$b_{k+1}-b_k={2H_k}/({k+1})$. Then, by Abel's lemma
$$
\sum_{k=1}^{n-1}(H_k^2-H_k^{(2)})  =-\sum_{k=1}^{n-1}(k+1)\frac{2H_k}{k+1}+n(H_n^2-H_n^{(2)})=n(H_n^2-H_n^{(2)})-2\sum_{k=1}^{n-1}H_k=n(H_n^2-H_n^{(2)})-2n(H_n-1).
$$
To prove (6), take $a_k=\binom{k}{2}$, so that $a_{k+1}-a_k=k$, and  $b_k=H_k^2-H_k^{(2)}$.
Then, again by Abel's lemma,
$$
\begin{array}{rl}
\displaystyle \sum_{k=1}^{n-1}k(H_k^2-H_k^{(2)}) & \displaystyle =-\sum_{k=1}^{n-1}\binom{k+1}{2}\frac{2H_k}{k+1}+\binom{n}{2}(H_n^2-H_n^{(2)})\\
& \displaystyle =\binom{n}{2}(H_n^2-H_n^{(2)})-2\sum_{k=1}^{n-1}kH_k=\binom{n}{2}(H_n^2-H_n^{(2)})-\frac{1}{4}n(n-1)(2H_{n}-1).
\end{array}
$$
\end{proof}

Let us consider now the following two  random variables:

\begin{itemize}

\item $S_n$, that chooses a tree $T\in \BT_n$ and computes  its \emph{Sackin index} \cite{Sackin:72}
$S(T)=\sum\limits_{i=1}^n\delta_T(i)$.

\item $S_n^{(2)}$, that chooses a tree $T\in \BT_n$ and computes
$S^{(2)}(T)=\hspace*{-2ex}\sum\limits_{1\leq i< j\leq n}\delta_T(i)^2$.
\end{itemize}

It is known that the expected value under the Yule model of $S_n$ is
$$
E_Y(S_n)=2n(H_n-1)\qquad\qquad  \mbox{\cite{KiSl:93}}.
$$ 
We shall compute now  the expected values under this model of $S_n^{(2)}$ and $D_n^{(2)}$: the first will be used in the computation of the second. To do this, we shall use the following recursive expressions for  $S^{(2)}(T\,\widehat{\ }\, T')$ and $D^{(2)}(T\,\widehat{\ }\, T')$.

\begin{nlemma}\label{rec}
Let $T,T'$ be two  phylogenetic  trees on disjoint sets of taxa $S,S'$, with $|S|=k$ and $|S'|=n-k$.
Then:
\begin{enumerate}[(1)]
\item  $S^{(2)}(T\,\widehat{\ }\, T')=S^{(2)}(T)+S^{(2)}(T')+ 2(S(T)+S(T'))+n$

\item $D^{(2)}(T\,\widehat{\ }\, T')  =D^{(2)}(T)+D^{(2)}(T')+
(n-k)(S^{(2)}(T)+4S(T)) +k(S^{(2)}(T')+4S(T'))+2S(T)S(T')+4k(n-k)$
\end{enumerate}
\end{nlemma}

\begin{proof}
Let us assume, without any loss of generality, that $S=\{1,\ldots,k\}$ and $S'=\{k+1,\ldots,n\}$ . 
Then, as far as (1) goes, we have that
$$
\delta_{T\,\widehat{\ }\, T'}(i)^2=\left\{
\begin{array}{ll}
(\delta_T(i)+1)^2 & \mbox{ if $1\leq i\leq k$}\\
(\delta_{T'}(i)+1)^2 & \mbox{ if $k+1\leq i\leq n$}
\end{array}
\right.
$$
and therefore
$$
\begin{array}{l}
S^{(2)}(T\,\widehat{\ }\, T') \displaystyle =\sum_{i=1}^{n}  \delta_{T\,\widehat{\ }\, T'}(i)^2=
\sum_{i=1}^{k}  (\delta_T(i)+1)^2+\sum_{i=k+1}^{n}  (\delta_{T'}(i)+1)^2\\
\quad \displaystyle=
\sum_{i=1}^{k}  (\delta_T(i)^2+2\delta_T(i)+1) +\sum_{i=k+1}^{n}  (\delta_{T'}(i)^2+2\delta_{T'}(i)+1)= S^{(2)}(T)+2S(T)+S^{(2)}(T')+2S(T')+n.
\end{array}
$$

As far as (2) goes, we have that
$$
d_{T\,\widehat{\ }\, T'}(i,j)^2=\left\{
\begin{array}{ll}
d_T(i,j)^2 & \mbox{ if $1\leq i<j\leq k$}\\
d_{T'}(i,j)^2 & \mbox{ if $k+1\leq i<j\leq n$}\\
(\delta_T(i)+\delta_{T'}(j)+2)^2 & \mbox{ if $1\leq i\leq k<j\leq n$}
\end{array}
\right.
$$
and therefore
$$
\begin{array}{l}
\displaystyle 
D^{(2)}(T\,\widehat{\ }\, T')=\hspace*{-2ex}\sum_{1\leq i< j\leq n}\hspace*{-2ex} d_{T\,\widehat{\ }\, T'}(i,j)^
2=
\hspace*{-2ex}\sum_{1\leq i< j\leq k}\hspace*{-2ex} d_{T}(i,j)^2+\hspace*{-2ex}\sum_{k+1\leq i< j\leq n}\hspace*{-2ex} d_{T'}(i,j)^2+\hspace*{-2ex}\sum_{1\leq i\leq k\atop k+1\leq j\leq n}\hspace*{-2ex} (\delta_T(i)+\delta_{T'}(j)+2)^2\\
\displaystyle\quad=
D^{(2)}(T)+D^{(2)}(T')
+
\hspace*{-2ex}\sum_{1\leq i\leq k\atop k+1\leq j\leq n}\hspace*{-2ex} (\delta_T(i)^2+\delta_{T'}(j)^2+2\delta_T(i)\delta_{T'}(j)+4\delta_T(i)+4\delta_{T'}(j)+4)\\
\displaystyle\quad=
D^{(2)}(T)+D^{(2)}(T')+(n-k)\sum_{i=1}^{k} (\delta_T(i)^2+4\delta_{T}(i))+k
\sum_{j=k+1}^{n} (\delta_{T'}(j)^2+4\delta_{T'}(j))\\
\displaystyle\quad\qquad \qquad \qquad +2\Big(\sum_{i=1}^k \delta_T(i)\Big)
\Big(\sum_{j=k+1}^{n}\delta_{T'}(j)\Big)+4k(n-k)\\
\displaystyle\quad=
D^{(2)}(T)+D^{(2)}(T')+
(n-k)(S^{(2)}(T)+4S(T))+k(S^{(2)}(T')+4S(T'))+2S(T)S(T')+4k(n-k).
\end{array}
$$
\end{proof}

Now we can compute explicit formulas for $E_Y(S_n^{(2)})$ and $E_Y(D_n^{(2)})$

\begin{nproposition}\label{EY:S2}
$E_Y(S_n^{(2)})= 4n(H_n^2-H_n^{(2)})-6n(H_n-1)$.
\end{nproposition}

\begin{proof}
We compute $E_Y(S_n^{(2)})$ using its very definition:
$$
\begin{array}{l}
E_Y(S_n^{(2)}) \displaystyle =\sum_{T\in \BT_n} S^{(2)}(T)\cdot p_Y(T)
 =\frac{1}{2}
\sum_{k=1}^{n-1}\sum_{S_k\subsetneq\{1,\ldots,n\}\atop |S_k|=k}
 \sum_{T_k\in \BT(S_k)}\sum_{T'_{n-k}\in \BT(S_k^c)}S^{(2)}(T_k\widehat{\ }\, {}T'_{n-k})\cdot p_Y(T_k\widehat{\ }\, {}T'_{n-k})
 \\ 
\quad   \displaystyle =\frac{1}{2} \sum_{k=1}^{n-1}\binom{n}{k}
 \sum_{T_k\in \BT_k}\sum_{T'_{n-k} \in \BT_{n-k}}\Big(S^{(2)}(T_k)+S^{(2)}(T_{n-k}')+ 2(S(T_k)+S(T_{n-k}'))+n\Big)\\
\displaystyle\quad\qquad 
\cdot\dfrac{2}{(n-1)\binom{n}{k}} P_Y(T_{k})P_Y(T'_{n-k})\\
\quad   \displaystyle =\frac{1}{n-1}\sum_{k=1}^{n-1} \bigg(
\sum_{T_k}\sum_{T'_{n-k}}  S^{(2)}(T_k)P_Y(T_{k})P_Y(T'_{n-k})+\sum_{T_k}\sum_{T'_{n-k}}  S^{(2)}(T'_{n-k})P_Y(T_{k})P_Y(T'_{n-k})\\
\displaystyle\quad\qquad 
+2\sum_{T_k}\sum_{T'_{n-k}}  S(T_k)P_Y(T_{k})P_Y(T'_{n-k})+2 \sum_{T_k}\sum_{T'_{n-k}}  S(T'_{n-k})P_Y(T_{k})P_Y(T'_{n-k})\\
\displaystyle\quad\qquad 
+ \sum_{T_k}\sum_{T'_{n-k}}  n P_Y(T_{k})P_Y(T'_{n-k})\bigg)\\
\quad   \displaystyle =\frac{1}{n-1}\sum_{k=1}^{n-1} \bigg(
\sum_{T_k}  S^{(2)}(T_k)P_Y(T_{k}) + \sum_{T'_{n-k}}  S^{(2)}(T'_{n-k})P_Y(T'_{n-k})\\
\displaystyle\quad\qquad 
+2 \sum_{T_k} S(T_k)P_Y(T_{k})+2\sum_{T'_{n-k}}  S(T_{n-k}')P_Y(T'_{n-k})+n\bigg)\\
\quad   \displaystyle =\frac{1}{n-1}\sum_{k=1}^{n-1} \bigg(
E_Y(S^{(2)}_k) + E_Y(S^{(2)}_{n-k})+2E_Y(S_k)+2E_Y(S_{n-k})+n\bigg)\\
\quad   \displaystyle =  
\frac{2}{n-1}\sum_{k=1}^{n-1}E_Y(S^{(2)}_k) +\frac{4}{n-1}\sum_{k=1}^{n-1}E_Y(S_k) +n
\end{array}
$$
In particular
$$
E_Y(S_{n-1}^{(2)})=\frac{2}{n-2}\sum_{k=1}^{n-2}E_Y(S^{(2)}_k) +\frac{4}{n-2}\sum_{k=1}^{n-2}E_Y(S_k) +n-1
$$
and therefore
$$
\begin{array}{l}
E_Y(S_n^{(2)}) \displaystyle =\frac{n-2}{n-1}\cdot \frac{2}{n-2}\sum_{k=1}^{n-2}E_Y(S^{(2)}_k)+\frac{2}{n-1}E_Y(S^{(2)}_{n-1})\\
\displaystyle\quad\qquad \quad\qquad 
+\frac{n-2}{n-1}\cdot \frac{4}{n-2}\sum_{k=1}^{n-2}E_Y(S_k)+\frac{4}{n-1}E_Y(S_{n-1})
+\frac{n-2}{n-1}\cdot (n-1)+2\\
\displaystyle \quad =
\frac{n-2}{n-1}E_Y(S_{n-1}^{(2)})+\frac{2}{n-1}E_Y(S^{(2)}_{n-1})+\frac{4}{n-1}E_Y(S_{n-1})+2=\frac{n}{n-1}E_Y(S_{n-1}^{(2)})+8H_{n-1}-6
\end{array}
$$
Setting $x_n=E_Y(S_n^{(2)})/n$, this recurrence becomes
$$
x_n=x_{n-1}+\frac{8H_{n-1}}{n}-\frac{6}{n}.
$$
Since $S^{(2)}$ applied to a single node is 0, $x_1=E_Y(S_1^{(2)})=0$, and the solution of this recursive equation with this initial condition is
$$
x_n=\sum_{k=2}^n\Big(\frac{8H_{k-1}}{k}-\frac{6}{k}\Big)=
8\sum_{k=1}^{n-1}\frac{H_k}{k+1}-6\sum_{k=2}^n \frac{1}{k}=
4(H_n^2-H_n^{(2)})-6(H_n-1)
$$
from where we deduce that
$$
E_Y(S_n^{(2)})=nx_n=4n(H_n^2-H_n^{(2)})-6n(H_n-1)
$$
as we claimed.
\end{proof}

\setcounter{theorem}{2}

\begin{theorem}
$E_Y(D_n^{(2)})= 8n(n+1)(H_{n}^2-H_{n}^{(2)})-2n(15n+7)H_n+45n^2-n$
\end{theorem}

\begin{proof}
Again, we compute $E_Y(D_n^{(2)})$ using its very definition:
$$
\begin{array}{l}
E_Y(D_n^{(2)}) \displaystyle =\sum_{T\in \BT_n} D^{(2)}(T)\cdot p_Y(T)
 \displaystyle =\frac{1}{2}
\sum_{k=1}^{n-1}\sum_{S_k\subsetneq\{1,\ldots,n\}\atop |S_k|=k}
 \sum_{T_k\in \BT(S_k)}\sum_{T'_{n-k}\in \BT(S_k^c)}D^{(2)}(T_k\widehat{\ }\, {}T'_{n-k})\cdot p_Y(T_k\widehat{\ }\, {}T'_{n-k}) \\ 

\quad   \displaystyle =\frac{1}{2} \sum_{k=1}^{n-1}\binom{n}{k}
 \sum_{T_k\in \BT_k}\sum_{T'_{n-k} \in \BT_{n-k}}\Big(D^{(2)}(T_k)+D^{(2)}(T'_{n-k})+2S(T_k)S(T'_{n-k})\\
\displaystyle\quad\qquad 
+
(n-k)(S^{(2)}(T_k)+4S(T_k))+k(S^{(2)}(T_{n-k}')+4S(T_{n-k}'))+4k(n-k)) \Big) 
\cdot\dfrac{2}{(n-1)\binom{n}{k}} P_Y(T_{k})P_Y(T'_{n-k})\\
\quad   \displaystyle =\frac{1}{n-1}\sum_{k=1}^{n-1} \bigg(
\sum_{T_k}\sum_{T'_{n-k}}  D^{(2)}(T_k)P_Y(T_{k})P_Y(T'_{n-k})+\sum_{T_k}\sum_{T'_{n-k}}  D^{(2)}(T'_{n-k})P_Y(T_{k})P_Y(T'_{n-k})\\
\displaystyle\quad\qquad 
+2\sum_{T_k}\sum_{T'_{n-k}} S(T_k)S(T'_{n-k})P_Y(T_{k})P_Y(T'_{n-k})+(n-k)\sum_{T_k}\sum_{T'_{n-k}}  S^{(2)}(T_k)P_Y(T_{k})P_Y(T'_{n-k})\\
\displaystyle\quad\qquad 
+4(n-k) \sum_{T_k}\sum_{T'_{n-k}}  S(T_k)P_Y(T_{k})P_Y(T'_{n-k})+k\sum_{T_k}\sum_{T'_{n-k}}  S^{(2)}(T_{n-k}')P_Y(T_{k})P_Y(T'_{n-k})\\
\displaystyle\quad\qquad 
+4k\sum_{T_k}\sum_{T'_{n-k}}  S(T_{n-k}')P_Y(T_{k})P_Y(T'_{n-k})+4\sum_{T_k}\sum_{T'_{n-k}}  k(n-k)P_Y(T_{k})P_Y(T'_{n-k})\bigg)\\

\quad   \displaystyle =\frac{1}{n-1}\sum_{k=1}^{n-1} \bigg(
\sum_{T_k}  D^{(2)}(T_k)P_Y(T_{k}) + \sum_{T'_{n-k}}  D^{(2)}(T'_{n-k})P_Y(T'_{n-k})\\
\displaystyle\quad\qquad 
+2\Big( \sum_{T_k} S(T_k)P_Y(T_{k})\Big)\Big(\sum_{T'_{n-k}}  S(T_{n-k}')P_Y(T'_{n-k})\Big)+(n-k)\sum_{T_k} S^{(2)}(T_k)P_Y(T_{k})\\
\displaystyle\quad\qquad 
+4(n-k) \sum_{T_k} S(T_k)P_Y(T_{k})+k \sum_{T'_{n-k}}  S^{(2)}(T_{n-k}') P_Y(T'_{n-k})+4k\sum_{T'_{n-k}}  S(T_{n-k}')P_Y(T'_{n-k})+4k(n-k)\bigg)\\

\quad   \displaystyle =\frac{1}{n-1}\sum_{k=1}^{n-1} \bigg(
E_Y(D^{(2)}_k) + E_Y(D^{(2)}_{n-k})+2E_Y(S_k)E_Y(S_{n-k})+(n-k)E_Y(S^{(2)}_k) \\
\displaystyle\quad\qquad 
+4(n-k)E_Y(S_k)+kE_Y(S^{(2)}_{n-k})+4kE_Y(S_{n-k})+4k(n-k)\bigg)\\
\quad   \displaystyle =\frac{2}{n-1}\sum_{k=1}^{n-1} \bigg(
E_Y(D^{(2)}_k) + E_Y(S_k)E_Y(S_{n-k})+(n-k)E_Y(S^{(2)}_k) +4(n-k)E_Y(S_k)\bigg)+\frac{2}{3}n(n+1)\\
\end{array}
$$
In particular
$$
\begin{array}{l}
E_Y(D_{n-1}^{(2)}) \displaystyle =\frac{2}{n-2}\sum_{k=1}^{n-2} \bigg(
E_Y(D^{(2)}_k) + E_Y(S_k)E_Y(S_{n-1-k})+(n-1-k)E_Y(S^{(2)}_k)\\
\displaystyle\quad\qquad 
 +4(n-1-k)E_Y(S_k)\bigg)+\frac{2}{3}n(n-1)
\end{array}
$$
and therefore
$$
\begin{array}{l}
E_Y(D_n^{(2)}) \displaystyle =
\frac{n-2}{n-1}\cdot \frac{2}{n-2}\sum_{k=1}^{n-2} E_Y(D^{(2)}_k)+\frac{2}{n-1} E_Y(D^{(2)}_{n-1})\\
\displaystyle\quad\qquad 
+\frac{n-2}{n-1}\cdot \frac{2}{n-2}\sum_{k=1}^{n-2}E_Y(S_k)E_Y(S_{n-1-k})+
\frac{2}{n-1}\bigg(\sum_{k=1}^{n-1}E_Y(S_k)E_Y(S_{n-k})-\sum_{k=1}^{n-2}E_Y(S_k)E_Y(S_{n-1-k})\Bigg)\\
\displaystyle\quad\qquad 
+\frac{n-2}{n-1}\cdot \frac{2}{n-2}\sum_{k=1}^{n-2}(n-1-k)E_Y(S^{(2)}_k)+
\frac{2}{n-1}\bigg(\sum_{k=1}^{n-1}(n-k)E_Y(S^{(2)}_k)-\sum_{k=1}^{n-2}(n-1-k)E_Y(S^{(2)}_k)\Bigg)\\
\displaystyle\quad\qquad 
+\frac{n-2}{n-1}\cdot \frac{8}{n-2}\sum_{k=1}^{n-2}(n-1-k)E_Y(S_k)+
\frac{8}{n-1}\bigg(\sum_{k=1}^{n-1}(n-k)E_Y(S_k)-\sum_{k=1}^{n-2}(n-1-k)E_Y(S_k)\Bigg)\\
\displaystyle\quad\qquad 
+\frac{n-2}{n-1}\cdot \frac{2}{3}n(n-1)+\frac{2}{3}n(n+1)-\frac{n-2}{n-1}\cdot \frac{2}{3}n(n-1)\\[1.5ex]

\quad \displaystyle =\frac{n-2}{n-1}E_Y(D^{(2)}_{n-1})+\frac{2}{n-1} E_Y(D^{(2)}_{n-1}) +\frac{2}{n-1}\sum_{k=1}^{n-2} E_Y(S_k)(E_Y(S_{n-k})-E_Y(S_{n-k-1}))\\
\displaystyle\quad\qquad 
+\frac{2}{n-1}\sum_{k=1}^{n-1} E_Y(S_{k}^{(2)})+\frac{8}{n-1}\sum_{k=1}^{n-1} E_Y(S_{k})+2n\\

\quad \displaystyle =\frac{n}{n-1}E_Y(D^{(2)}_{n-1})+\frac{8}{n-1}\sum_{k=1}^{n-2} k(H_k-1) H_{n-k-1}+\frac{2}{n-1}\sum_{k=1}^{n-1} (4k(H_k^2-H_k^{(2)})-6k(H_k-1))\\
\displaystyle\quad\qquad 
+\frac{16}{n-1}\sum_{k=1}^{n-1} k(H_k-1)+2n\\

\quad \displaystyle =\frac{n}{n-1}E_Y(D^{(2)}_{n-1})+\frac{8}{n-1}\sum_{k=1}^{n-2} kH_kH_{n-k-1}-
\frac{8}{n-1}\sum_{k=1}^{n-2} kH_{n-k-1}+\frac{8}{n-1}\sum_{k=1}^{n-1}  k(H_k^2-H_k^{(2)})\\
\displaystyle\quad\qquad 
+\frac{4}{n-1}\sum_{k=1}^{n-1} k(H_k-1)+2n\\

\quad \displaystyle =\frac{n}{n-1}E_Y(D^{(2)}_{n-1})+\frac{8}{n-1}\sum_{k=1}^{n-2} kH_kH_{n-k-1}
-\frac{8}{n-1}\sum_{k=1}^{n-2} (n-k-1)H_k+\frac{8}{n-1}\sum_{k=1}^{n-1}  k(H_k^2-H_k^{(2)})\\
\displaystyle\quad\qquad 
+\frac{4}{n-1}\sum_{k=1}^{n-1} kH_k-\frac{4}{n-1}\sum_{k=1}^{n-1} k+2n\\

\quad \displaystyle =\frac{n}{n-1}E_Y(D^{(2)}_{n-1})+\frac{8}{n-1}\sum_{k=1}^{n-2} kH_kH_{n-k-1}
-8\sum_{k=1}^{n-2} H_k
+\frac{8}{n-1}\sum_{k=1}^{n-1}  k(H_k^2-H_k^{(2)})\\
\displaystyle\quad\qquad 
+\frac{12}{n-1}\sum_{k=1}^{n-1} kH_k-8H_{n-1}\\

\quad \displaystyle =\frac{n}{n-1}E_Y(D^{(2)}_{n-1})+ 4n(H_{n}^2-H_{n}^{(2)}-2H_{n}+2)
-8(n-1)(H_{n-1}-1)+4n(H_n^2-H_n^{(2)})\\
\displaystyle\quad\qquad 
-2n(2H_{n}-1)+3n(2H_{n}-1)-8H_{n-1}\\

\quad \displaystyle =\frac{n}{n-1}E_Y(D^{(2)}_{n-1})+
8n(H_{n}^2-H_{n}^{(2)})-14nH_{n-1}+15n-14
\end{array}
$$
Setting $x_n=E_Y(D_n^{(2)})/n$, this recurrence becomes
$$
x_n=x_{n-1}+8(H_{n}^2-H_{n}^{(2)})-14H_{n-1}+15-\frac{14}{n}.
$$
The solution of this recursive with $x_1=E_Y(D_1^{(2)})=0$ is
$$
\begin{array}{rl}
x_n& \displaystyle=\sum_{k=2}^{n}\Big(8(H_{k}^2-H_{k}^{(2)})-14H_{k-1}+15-\frac{14}{k}\Big)\\
& \displaystyle=8\sum_{k=1}^{n} (H_{k}^2-H_{k}^{(2)})-14\sum_{k=1}^{n-1}H_{k}+15(n-1)-14\sum_{k=2}^{n}\frac{1}{k}\\
& \displaystyle=8(n+1)(H_{n+1}^2-H_{n+1}^{(2)})-16(n+1)(H_{n+1}-1)-14n(H_n-1)+15(n-1)-14(H_n-1)\\
& \displaystyle=8(n+1)(H_{n}^2-H_{n}^{(2)})-2(15n+7)H_n+45n-1
\end{array}
$$
from where we deduce that
$$
E_Y(D_n^{(2)})=nx_n=8n(n+1)(H_{n}^2-H_{n}^{(2)})-2n(15n+7)H_n+45n^2-n
$$
as we claimed.
\end{proof}

\end{document}